\documentclass{article}

\usepackage[a4paper]{geometry}

\usepackage[dvipdfmx]{graphicx}
\usepackage{amsmath, amssymb, amsfonts}
\usepackage{mathtools}
\usepackage{amsthm}
\usepackage{algorithm, algpseudocode}
\usepackage{natbib}
\usepackage{comment}
\usepackage{array}
\usepackage{dcolumn}
\usepackage{bm}

\newtheorem{thm}{Theorem}
\newtheorem{lem}{Lemma}

\newtheorem{prop}{Proposition}

\theoremstyle{definition}
\newtheorem{defi}{Definition}
\newtheorem{ex}{Example}

\newtheorem{ass}{Assumption}

\newcommand{\abf}{\mathbf{a}}
\newcommand{\bbf}{\mathbf{b}}
\newcommand{\Fbb}{\mathbb{F}}
\newcommand{\Nbb}{\mathbb{N}}

\newcommand{\Rbb}{\mathbb{R}}
\newcommand{\Ubb}{\mathbb{U}}
\newcommand{\Zbb}{\mathbb{Z}}
\newcommand{\Ccal}{\mathcal{C}}
\newcommand{\Fcal}{\mathcal{F}}
\newcommand{\Pcal}{\mathcal{P}}
\newcommand{\Zcal}{\mathcal{Z}}
\newcommand{\ind}{\mathbf{1}}
\newcommand{\obs}{\mathrm{obs}}
\newcommand{\pval}{\mathrm{pval}}
\newcommand{\Fbin}{F_\mathrm{bin}}
\newcommand{\T}{\mathsf{T}}
\newcommand{\wtld}[1]{\widetilde{#1}}
\newcommand{\Exp}{\mathrm{E}}
\newcommand{\Var}{\mathrm{V}}
\newcommand{\Cov}{\mathrm{Cov}}
\newcommand{\ave}{\mathrm{ave}}
\newcommand{\tr}{\mathrm{tr}}

\mathtoolsset{showonlyrefs=true,showmanualtags=true}

\title{Improving Randomization Tests under Interference Based on Power Analysis}
\author{Mizuho Yanagi\thanks{Graduate School of Information Science and Technology,
The University of Tokyo, 7-3-1 Hongo, Bunkyo-ku, Tokyo, 113-8656, Japan.} \and Tomonari Sei\footnotemark[1] \thanks{sei@mist.i.u-tokyo.ac.jp}}

\begin{document}

\maketitle

\begin{abstract}
%\begin{comment}
%In causal inference, we can consider a situation in which treatment on one unit affects others, i.e., interference exists.
%In the presence of interference, we cannot perform a classical randomization test directly because a null hypothesis is not sharp.
%Instead, we need to perform the randomization test restricted to a subset of units and assignments that makes the null hypothesis sharp.
%\cite{Puelz} used a graphical approach to construct a generic test method, a biclique test, by reducing the selection of the appropriate subsets to searching for bicliques in a bipartite graph.
%However, since the power depends on the features of selected subsets, there is still room to improve the power by refining the selection procedure.
%In this paper, we propose a method to improve the power of the biclique test based on a power evaluation of the randomization test.
%We explicitly derived a power expression for the randomization test under several assumptions and found that a certain quantity calculated from a given assignment set characterizes the power.
%Based on this fact, we proposed a method to improve the power of the biclique test by modifying the selection rule for subsets of units and assignments.
%Through a simulation with a spatial interference setting, we confirmed that the proposed method has higher power than the existing method.
%\end{comment}
In causal inference, we can consider a situation in which treatment on one unit affects others, i.e., interference exists.
In the presence of interference, we cannot perform a classical randomization test directly because a null hypothesis is not sharp.
Instead, we need to perform the randomization test restricted to a subset of units and assignments that makes the null hypothesis sharp.
A previous study constructed a useful testing method, a biclique test, by reducing the selection of the appropriate subsets to searching for bicliques in a bipartite graph.
However, since the power depends on the features of selected subsets, there is still room to improve the power by refining the selection procedure.
In this paper, we propose a method to improve the biclique test based on a power evaluation of the randomization test.
We explicitly derived an expression for the power of the randomization test under several assumptions and found that a certain quantity calculated from a given assignment set characterizes the power.
Based on this fact, we propose a method to improve the power of the biclique test by modifying the selection rule for subsets of units and assignments.
Through a simulation with a spatial interference setting, we confirm that the proposed method has higher power than the existing method.
\end{abstract}

\section{Introduction}
In causal inference, we usually assume that the outcome of each unit depends only on treatment for itself.
This assumption is called {\it no interference} \citep{Cox}, which is one of the components of the stable unit treatment value assumption (SUTVA) \citep{Rubin}.
However, in practice, there are many situations where interference exists.
For example, we expect that vaccination to a part of a population reduces the number of infected units and consequently lowers the infection risk for untreated units, which is known as herd immunity.
It is one of the examples where interference does exist.
In recent years, causal inference under interference has received much attention and has been studied from various approaches (see \cite{Halloran_and_Hudgens} for a review).

In the presence of interference, standard approaches often fail.
One example is the {\it randomization test}.
The randomization test is a classical method proposed by \cite{Fisher} to test a sharp null hypothesis of no treatment effect.
However, when interference exists, we cannot perform it because a null hypothesis is no longer sharp.
Hence, alternative methods were proposed to test a non-sharp null hypothesis \citep{Aronow, Athey, Basse, Puelz}.
The main idea of these methods is to focus on a subset of units and assignments on which the null hypothesis is sharp and perform the randomization test restricted on it.
These methods are called {\it conditional randomization tests} because we perform them conditionally on a selected subset.

What characterizes the conditional randomization test is how to select the appropriate subset, which determines the power.
\cite{Athey} took the approach of selecting focal units regardless of a realized assignment.
Although it is generally applicable, the power may be low because it does not use any information about the realized assignment.
Later, \cite{Basse} showed a general procedure of the conditional randomization test, which incorporates the information of the realized assignment.
However, it did not clarify the concrete procedure for selecting subsets and its applications were limited.

On the other hand, \cite{Puelz} took a graphical approach to construct a concrete procedure for selecting a subset.
They showed that appropriate subsets of units and assignments correspond to bicliques in a bipartite graph constructed for a null hypothesis and reduced the problem to searching for bicliques in the graph.
Based on this idea, they proposed {\it biclique tests}, which can be applied to any type of interference.

In this paper, we propose a method to improve the biclique test.
Since the power of the biclique test depends on the features of selected bicliques, we can improve it by modifying the procedure to select more desirable bicliques.
Thus, we analyzed the randomization test to clarify what features characterize its power and proposed a method to improve the biclique test by modifying the selection rule to get bicliques.
Also, we compared the proposed method with the existing method through simulations in scenarios with spatial interference.

This paper is organized as follows.
In Section \ref{Puelz's method}, we describe the formulation of null hypotheses under interference and review the previous studies.
In Section \ref{proposed method}, we propose our method to improve the biclique test.
In Section \ref{numerical experiment}, we show the numerical experiments.
In Section \ref{conclusion}, we conclude with a summary and future perspectives.
%Supplementary material includes the proofs of theorems and propositions.
Appendix includes the proofs of theorems and propositions.

%%%%%%%%%%%%%%%%%%%
\section{Conditional randomization tests by graph-theoretic approach: a review}
\label{Puelz's method}

\subsection{Notation and formulation of null hypotheses}
Let $\Ubb = \{1, \dots, N\}$ denote the set of all units and let $z = (z_1, \dots, z_N)^{\T} \in \{0, 1\}^{N}$ denote the vector of treatment assignments for the $N$ units ($1$ : treatment, $0$ : control).
The treatment assignment is selected according to a known probability distribution $P(z)$, which is determined by experimenters.
The support of $P(z)$, i.e., the set of possible assignments, is denoted by $\Zbb = \{z \in \{0,1\}^N : P(z) > 0\}$ and the realized assignments are denoted by $Z^{\obs}$.
Let $Y(z) = (Y_1(z), \dots, Y_N(z))^{\T} \in \Rbb^N$ be the vector of potential outcomes for the units and denote the outcomes realized under the actual assignment by $Y^{\obs} = Y(Z^{\obs})$.
Note that $Y(z)$ is a fixed vector and the stochastic behavior in this setting is caused by the randomization of assignments $Z^\obs\sim P(z)$.

%(spilloverの説明)
Following the framework of \cite{Manski} and \cite{Aronow_and_Samii}, we formulate interference.
Now, we define a map $f_i : \Zbb \rightarrow \Fbb$, where $\Fbb$ is a suitable (finite) set.
We assume that each unit $i$ is exposed to $f_i(z) = \abf \in \Fbb$ under the assignment $z$.
This $\abf$ is a low-dimensional summary of $z$ that represents an intrinsic exposure to the unit $i$.
We call it a {\it treatment exposure}.
The map $f_i$ characterizes the type of interference between units and we call it an {\it exposure mapping function}.
For example, under the assumption of no interference, we can write $f_i(z) = z_i$ simply.
In this case, we denote $Y_i(z)$ as $Y_i(z_i)$.

\begin{ex}[Spatial interference]
\label{spatial}
Suppose that each unit $i$ is located at $x_i \in \Rbb^2$.
Considering a situation where spatial interference exists, the set of treatment exposures and the exposure mapping function are
\begin{align}
	\Fbb = \{0, 1, 2\}, \quad f_i(z) = 
	\begin{cases}
		2 & (z_i = 1) \\
		1 & (z_i = 0 \text{ and } z_j = 1, ||x_i - x_j|| \leq r \text{ for some $j \in \Ubb$}) \\
		0 & (\text{otherwise})
	\end{cases}
\end{align}
respectively.
In this setting, we assume that each unit is affected by other units being treated within a certain distance $r$ from it, in addition to whether or not being treated itself.
\end{ex}

In the presence of interference, a null hypothesis of interest is whether the potential outcomes vary under different treatment exposures.
We can express this hypothesis in the following general form:
\[
	H^{\Fcal}_0 : Y_i(z) = Y_i(z') \ \text{for any} \ (i, z, z') \in \Ubb \times \Zbb^2 \ \text{s.t.} \ f_i(z), f_i(z') \in \Fcal,
\]
where $\Fcal \subset \Fbb$ characterizes the null hypothesis.
\begin{comment}
As a special case, when $\Fcal = \Fbb$,
\[
	H^{\Fbb}_0 : Y_i(z) = Y_i(z') \ \text{for any} \ (i, z, z') \in \Ubb \times \Zbb^2.
\]
This represents the hypothesis that outcomes do not change under any assignment, i.e., there is no effect of the treatment.
Also, if $\Fcal = \{\abf\}$, 
\[
	H^{\{\abf\}}_0 : Y_i(z) = Y_i(z') \ \text{for any} \ (i, z, z') \in \Ubb \times \Zbb^2 \ \text{s.t.} \ f_i(z) = f_i(z') = \abf,
\]
which expresses the hypothesis that the outcome is always the same under the assignment when the exposure is $\abf$.
\end{comment}
For example, when $\Fcal = \{{\bf a, b}\}$, we can write
\[
	H^{\{{\bf a, b}\}}_0 : Y_i(z) = Y_i(z') \ \text{for any} \ (i, z, z') \in \Ubb \times \Zbb^2 \ \text{s.t.} \ f_i(z), f_i(z') \in \{{\abf, \bbf}\},
\]
which indicates that there is no difference in the effect between two treatment exposures $\abf$ and $\bbf$.
It is called a {\it contrast hypothesis} and we will mainly focus on this type of hypothesis in this paper.

\begin{comment}
For each of the above examples of interference between units, examples of the null hypothesis of interest are given below.

\begin{ex}[Clustered interference (cont.)]
The hypothesis of interest is whether the outcome for each unit is affected by the treatment of the other units.
That is, whether the value of the outcome is different for treatment exposures 0 and 1, which corresponds to the null hypothesis $H^{\{0,1\}}_0$.
\end{ex}
\end{comment}

\setcounter{ex}{1}
\begin{ex}[Spatial interference (cont.)]
One hypothesis of interest is $H^{\{0, 1\}}_0$.
This hypothesis corresponds to the question of whether untreated units are affected by the presence of treated units in their neighborhood.
\end{ex}

\subsection{Conditional randomization tests}
Before discussing tests of the null hypothesis under interference $H^{\Fcal}_0$, we review the classical randomization test.
Under no interference, we test a hypothesis of no treatment effect, 
\begin{align}
	H_0 : Y_i(0) = Y_i(1) \ (i \in \Ubb).
	\label{eq: H_0}
\end{align}
The randomization test \citep{Fisher} is a method to test this null hypothesis.

Let $T(z, Y)$ be a test statistic for $z \in \Zbb$ and $Y \in \Rbb^{N}$.
For example, we can use the difference in means
\[
	T(z, Y) = \ave\{Y_i \ | \ i \in \Ubb, \ z_i = 1\}
		- \ave\{Y_i \ | \ i \in \Ubb, \ z_i = 0\},
\]
where $\ave\{\bullet\}$ denotes the average of the elements of the set.
In the randomization test, we compute the exact p-value as the probability that the test statistic $T(Z, Y(Z))$ is greater than or equal to the observed value $T^{\obs} = T(Z^{\obs}, Y(Z^{\obs}))$ under the null hypothesis $H_0$:
\[
	\pval(Z^{\obs})
	= \Exp_Z\left[\ind\{T(Z, Y(Z)) \geq T^{\obs})\} | H_0\right].
\]
To compute it, we need to obtain the distribution of the test statistic under the null hypothesis.
Now, we have $Y(z) = Y^{\obs}$ for any assignment $z$ under $H_0$.
Therefore,  
\[
	T(Z, Y(Z)) = T(Z, Y^{\obs})
\]
holds and we obtain the distribution of the test statistic $T$ induced by $P(z)$.
The procedure of the randomization test is summarized as follows.

\begin{thm}[Randomization test]
\label{RT}
Let $H_0$ be a null hypothesis.
The p-value obtained from the following procedure
\begin{enumerate}
	\item Draw $Z^{\obs} \sim P(z)$, and obsereve $Y^{\obs} = Y(Z^{\obs})$.
	\item Compute $T^{\obs} = T(Z^{\obs}, Y^{\obs})$.
	\item Compute $\pval(Z^{\obs}) 
		= \Exp_Z\left[\ind\{T(Z, Y^{\obs}) \geq T^{\obs}\} \right]$.
\end{enumerate}
is valid.
That is, under $H_0$ for any $\alpha \in (0, 1)$,
\[
	P(\pval(Z^\obs) \leq \alpha) \leq \alpha
\]
holds.
\end{thm}

Under the null hypothesis \eqref{eq: H_0}, we can infer the potential outcomes for unobserved assignments.
This property enables us to compute the exact p-value.
The null hypothesis that satisfies the property is said to be {\it sharp}.

\begin{defi}[Sharp null hypothesis]
\label{sharp null hypothesis}
A null hypothesis $H_0$ is sharp if we can infer $Y(z)$ for any $z \in \Zbb$ from $(Z^{\obs}, Y^{\obs})$ under $H_0$.
\end{defi}

Back to the test of the null hypothesis under interference $H^\Fcal_0$.
Since the null hypothesis $H^{\Fcal}_0$ is usually not sharp, we cannot perform the classical randomization test.
A possible alternative is to focus on a subset of units and assignments, on which $H^{\Fcal}_0$ is sharp, and perform the randomization test restricted on it.
We define sharpness by the restriction as follows.

\begin{comment}
\setcounter{thm}{1}
\begin{ex}[clustered interference (cont.)]
The null hypothesis $H^{\{0, 1\}}_0$ is not sharp.
Actually, when $Z^{\obs}_i = 0$, we cannot infer $Y_i(z)$ for assignment $z$ such that $z_i = 1$ under $H^{\{0, 1\}}_0$.
\end{ex}
\setcounter{thm}{5}
\end{comment}

\begin{defi}[Sharp null hypothesis on $C$]
Let $C =(U, \Zcal)$ be a pair of a subset of units and assignments, $U \subseteq \Ubb, \Zcal \subseteq \Zbb$.
A null hypothesis $H_0$ is sharp on $C$ if we can infer $Y_U(z)$ for any $z \in \Zcal$ from $(Z^{\obs}, Y^{\obs})$ under $H_0$.
Here, $Y_U(z)$ denotes the sub-vector of $Y(z)$ corresponding to a subset $U$.
\end{defi}

The standard definition of sharpness (Definition \ref{sharp null hypothesis}) corresponds to the case where $C = (\Ubb, \Zbb)$.
Here, the subset $C = (U, \Zcal)$ is called a {\it conditioning event} and the procedure for selecting a conditioning event $C$ from $Z^{\obs}$ is called a {\it conditioning mechanism}, which we denote by $P(C|Z^{\obs})$ since it is a probablistic procedure in general.
The conditioning mechanism is assumed to satisfy $z \in \Zcal$ if $P(C|z) > 0$, so that we select a subset $C$ containing $Z^{\obs}$.

Since we perform the randomization test restricted on $C$, the test statistic $T$ must depend only on outcomes of units in $U$.
For example, when testing the contrast hypothesis $H^{\{\abf, \bbf\}}_0$, we can use the difference in means type test statistic:
\begin{align}
	T(z, Y; C) = \ave\{Y_i \ | \ i \in U, \ f_i(z) = \abf\}
		- \ave\{Y_i \ | \ i \in U, \ f_i(z) = \bbf\}.
\end{align}
We call such $T$ a {\it restricted test statistic}.

\begin{defi}[Restricted test statistic]
Let $C = (U, \Zcal)$ be some conditioning event.
A test statistic $T$ is said to be restricted on $C$ if $T$ satisfies
\[
	T(z, Y; C) = T(z, Y'; C) \ \text{for any} \ z \in \Zbb \ \text{and} \ \text{any} \ Y, Y' \in \Rbb^N \ \text{such that} \ Y_U = Y'_U.
\]
\end{defi}

Based on the above, we show the procedure of the conditional randomization test.

\begin{thm}[Conditional randomization test \citep{Basse}]
\label{conditional RT}
Let $H_0$ be a null hypothesis.
For each conditioning event $C = (U, \Zcal)$, let $T(z, Y; C)$ be a test statistic restricted on $C$.
Suppose that a conditioning mechanism $P(C|z)$ satisfies
\begin{align}
	Y_U(z) = Y_U(z') \ \text{for any} \ z, z' \in \Zbb \ \text{s.t.} \ P(C|z) > 0, P(C|z') > 0
	\label{eq: sharp}
\end{align}
for any $C$ under $H_0$.
Then, the p-value obtained from the following procedure
\begin{enumerate}
	\item Draw $Z^\obs \sim P(z)$, and observe $Y^\obs = Y(Z^\obs)$.
	\item Draw $C \sim P(C|Z^\obs)$.
	\item Compute $T^\obs = T(Z^\obs, Y^\obs; C)$.
	\item Compute $\pval(Z^\obs; C)
		= \Exp_{Z}\left[\ind\{T(Z, Y^\obs; C) \geq T^\obs\} | C\right]$.
\end{enumerate}
is conditionally valid.
That is, under $H_0$ for any $\alpha \in (0, 1)$,
\[
	P(\pval(Z^\obs; C) \leq \alpha | C) \leq \alpha
\]
holds.
Here, the expectation in step (d) is taken for $P(z|C) \propto P(z) P(C|z)$.
\end{thm}

Note that conditional validity of the p-value means marginal validity in the following sense:
\begin{align}
	\sum_C P(\pval(Z^\obs; C) \leq \alpha | C) P(C) \leq \sum_C \alpha P(C) = \alpha.
\end{align}

The difference from the classical randomization test (Theorem \ref{RT}) is in steps (b) and (d).
In step (b), for the observed $Z^\obs$, we select an conditioning event $C$ according to $P(C|z)$.
The null hypothesis becomes sharp on $C$, which we select according to the conditioning mechanism satisfying \eqref{eq: sharp}.
Then, in step (d), we calculate the p-value based on the conditional distribution $P(z|C) \propto P(z) P(C|z)$.
This procedure allows us to calculate the exact p-value even for the non-sharp null hypothesis in the classical sense (Definition \ref{sharp null hypothesis}).
Note that the classical randomization test corresponds to the conditioning mechanism $P(C = (\Ubb, \Zbb)|z) = 1 \ (\text{for all $z \in \Zbb$})$ as a special case.

The conditional randomization test is characterized by the conditioning mechanism $P(C|z)$, which determines the power of the test.
First of all, $P(C|z)$ must satisfy the following two requirements so that a conditional randomization test is feasible \citep{Puelz}.
First, $P(C|z)$ must satisfy \eqref{eq: sharp}, which is necessary for the null hypothesis to be sharp.
Second, we must be able to sample $z$ from $P(z|C) \propto P(C|z)P(z)$.
It is necessary to compute the p-value in step (d) by Monte Carlo approximation when the support of $P(z|C)$ is too large to compute it exactly.

%%%%%%%%%%%%%%%%%%%

\subsection{Biclique tests}
Theorem \ref{conditional RT} represents the general procedure for the conditional randomization test, but it does not reveal how to construct an appropriate conditioning mechanism $P(C|z)$.
\cite{Puelz} proposed a conditioning mechanism that can be applied generally by a graph-theoretic approach.
The concept playing a central role in their method is a {\it null exposure graph}.

\begin{defi}[Null exposure graph]
For a null hypothesis $H^{\Fcal}_0$ corresponding a exposure mapping function $f_i(z)$ and a subset of treatment exposures $\Fcal \subseteq \Fbb$, we define a bipartite graph $G^{\Fcal}_{f} = (V, E)$ as
\begin{align}
	V = \Ubb \cup \Zbb, \quad E = \{(i, z) \in \Ubb \times \Zbb \ | \ f_i(z) \in \Fcal \}.
\end{align}
The bipartite graph $G^{\Fcal}_{f}$ is a null exposure graph of $H^{\Fcal}_0$.
\end{defi}

Next, we define a complete sub-bipartite graph in a null exposure graph as a {\it biclique}.

\begin{defi}[Biclique]
A biclique of a null exposure graph $G^{\Fcal}_{f} = (V, E)$ is a pair of subsets of vertices $C = (U, \Zcal) \ (U \subseteq \Ubb, \Zcal \subseteq \Zbb)$ that satisfies
\[
	(i, z) \in E \quad (\text{for all $i \in U$ and $z \in \Zcal$}).
\]
\end{defi}

\begin{comment}
\setcounter{section}{2}
\setcounter{thm}{1}
\begin{ex}[clustered interference (cont.)]
For $\Ubb = \{1, 2, 3, 4\}$, consider the case in which there are two groupus, $\{1, 2\}$ and $\{3, 4\}$.
We have 4 differents assignments each of which gives treatment for only one unit, i.e., $\Zbb = \{z^{(1)}, z^{(2)}, z^{(3)}, z^{(4)}\}$ (e.g., $z^{(3)}=(0, 0, 1, 0)^\T$).
In this case, the null exposure graph $G^\Fcal_f$ corresponding to the null hypothesis $H^{\{0, 1\}}_0$ is givien in Figure \ref{biclique}.
Also, $C = (\{1, 3\}, \{z^{(2)}, z^{(4)}\})$ is one of the bicliques of $G^\Fcal_f$ (shown in red in the figure).
\setcounter{section}{3}
\setcounter{thm}{2}

%\begin{figure}[htbp]
%	\begin{center}
%		\includegraphics[width=8cm]{biclique.jpg}
%		\caption{Example of null exposure graph and biclique.
%		The whole bipartite graph represents the null exposure graph corresponding to the null hypothesis $H^{\{0, 1\}}$ in Example \ref{cluster}, and the red bipartite subgraph represents one of the bicliques.}
%		\label{biclique}
%	\end{center}
%\end{figure}

\end{ex}
\end{comment}

Bicliques of the null exposure graph have the following important meaning.
\begin{prop}
Let $C = (U, \Zcal)$ be a biclique of a null exposure graph $G^{\Fcal}_f$ of a null hypothesis $H^{\Fcal}_0$.
If $Z^{\obs} \in \Zcal$, then
\[
	Y_i(z) = Y_i(Z^\obs) \quad (\text{for all $i \in U$ and $z \in \Zcal$})
\]
holds under $H^{\Fcal}_0$.
\end{prop}

\begin{proof}
Let $i \in U$ be any unit in the biclique $C = (U, \Zcal)$.
Then, $f_i(Z^\obs) \in \Fcal$ holds by the definition of the null exposure graph.
Similarly, $f_i(z) \in \Fcal$ holds for any $z \in \Zcal$.
Therefore, $Y_i(z) = Y_i(Z^\obs)$ holds under $H^{\Fcal}_0$ by its definition.
\end{proof}

This proposition states that the null hypothesis $H^{\Fcal}_0$ is sharp on a biclique $C$.
Hence, for an observed assignment $Z^\obs$, it is an appropriate conditioning mechanism to choose a biclique $C$ containing $Z^\obs$ as a conditioning event.

Since there are usually several bicliques containing a given $Z^\obs$, we have a choice on which biclique to select.
Thus, we partition the null exposure graph into several bicliques in advance, and select a biclique based on the partition, which we call a {\it biclique decomposition}.

\begin{defi}[Biclique decomposition]
For a null exposure graph $G^{\Fcal}_f$, a biclique decomposition of $G^{\Fcal}_f$ is a set of bicliques $\Ccal = \{C_1, \dots, C_K\} \ (C_k = (U_k, \Zcal_k)\ (k = 1, \dots, K))$ that satisfies
\[
	\bigcup_{k = 1}^{K} \Zcal_k = \Zbb, \quad \Zcal_k \cap \Zcal_{k'} = \emptyset \ (k \neq k').
\]
\end{defi}
Given a biclique decomposition $\Ccal$, there is a unique biclique $C \in \Ccal$ that contains $Z^{\obs}$.
Therefore, we can consider the conditioning mechanism that we select such a $C$ as a conditioning event decisively.
This procedure is written explicitly as $P(C|z) = \ind \{z \in \Zcal(C)\}$.
Here, $\Zcal(C)$ denotes the set of assignments corresponding to a biclique $C$.

Based on the above, we show a procedure of conditional randomization tests based on a biclique decomposition, which we call {\it biclique tests}.

\begin{thm}[Biclique test \citep{Puelz}]
Let $H^{\Fcal}_0$ be a null hypothesis.
For each conditioning event $C = (U, \Zcal)$, let $T(z, Y; C)$ be a test statistic restricted on $C$.
Suppose that a biclique decomposition $\Ccal$ of the null exposure graph $G^{\Fcal}_f$ corresponding to $H^{\Fcal}_{0}$ is given.
Then, the p-value obtained from the following procedure
\begin{enumerate}
	\item Draw $Z^\obs \sim P(z)$, and observe $Y^\obs = Y(Z^\obs)$.
	\item Find the unique biclique $C = (U, \Zcal)\in \Ccal$ such that $Z^\obs \in \Zcal$.
	\item Compute $T^\obs = T(Z^\obs, Y^\obs; C)$.
	\item Compute $\pval(Z^\obs; C)
		= \Exp_{Z \sim r}\left[\ind\{T(Z, Y^\obs; C) \geq T^\obs\} \right]$.
\end{enumerate}
is conditionally valid.
That is, under $H_0$ for any $\alpha \in (0, 1)$,
\[
	P(\pval(Z^\obs; C) \leq \alpha | C) \leq \alpha
\]
holds.
Here, the expectation in step (d) is taken for $r(z) \propto P(z) \ind \{z \in \Zcal(C)\}$.
\end{thm}

\begin{comment}
Figure \ref{biclique test fig} shows the procedure of the biclique test.
In advance, we construct the null exposure graph corresponding to the null hypothesis (left figure) and get the biclique decomposition $\Ccal$ (middle figure).
After $Z^\obs$ is observed by randomization, we perform the classical randomization test on the subset of units and assignments $\Ccal = (U, \Zcal)$, based on the distribution $r(z)$ which induced by restricting the support of $P(z)$ on $\Zcal$ (right figure).

Usually, we take $P(z)$ as a uniform distribution on $\Zbb$.
Then, the induced distribution $r(z)$ is also uniform on $\Zcal$, and we can easily calculate the p-value.
\end{comment}

%\begin{figure}[htbp]
%	\begin{center}
%		\includegraphics[width=12cm]{biclique_test.jpg}
%		\caption{The figure of the procedure of the biclique test.
%		The biclique test consists of three steps: constructing the null exposure graph (left figure), constructing the biclique decomposition (middle figure), and running the conditional randomization test (right firgure).}
%		\label{biclique test fig}
%	\end{center}
%\end{figure}

What type of biclique decomposition is desirable to perform the biclique test?
Now, letting us denote the test function by $\phi(z)$, the power of the biclique test is expressed as
\[
	\Exp_Z[\phi(Z)] = \Exp_C[\Exp_Z[\phi(Z) | C]]
	= \Exp_C[\Exp_{Z \sim r}[\phi(Z)]].
\]
Here, the expectation on $C$ is taken for $P(C) = \sum_{z \in \Zcal(C)} P(z)$.
Thus, the power of the biclique test is the average of $\Exp_{Z \sim r}[\phi(Z)]$, which is the power of the randomization test for each biclique in the biclique decomposition.
Therefore, to attain high power in the biclique test, it is important to decompose the null exposure graph so that the randomization test in each biclique has high power.

\cite{Puelz} evaluated the power of the randomization test on the biclique $C = (U, \Zcal)$ under certain assumptions: the larger the number of units and assignments, the higher the power (\cite{Puelz} Theorem 3).
According to this evaluation, it is desirable to obtain a biclique decomposition such that the size of each biclique is as large as possible.
We can obtain such a biclique decomposition by the following greedy method (Algorithm \ref{biclique decomposition}).
Let $E(G), U(G), \Zcal(G)$ be the set of edges, units, and assignments corresponding to a graph $G$, respectively.
Also, $C \in G$ denotes that $C$ is a biclique of $G$.

\begin{algorithm}[H]
\caption{Biclique decomposition algorithm}
\label{biclique decomposition}
\begin{algorithmic}[1]
	\Require null exposrure graph $G^{\Fcal}_{f}$
	\Ensure biclique decomposition $\Ccal$
	\State $\Ccal \leftarrow \emptyset, \ G \leftarrow G^{\Fcal}_{f}$ 
	\While{$|\Zcal(G)| > 0$}
		\State $C^{\ast} = \mathrm{argmax}_{C \in G} |E(C)|$
		\State $E(G) \leftarrow E(G) \setminus E(C^{\ast}), \ 
			\Zcal(G) \leftarrow \Zcal(G) \setminus \Zcal(C^{\ast})$
		\State $\Ccal \leftarrow \Ccal \cup \{C^{\ast}\}$
	\EndWhile
	\State \Return $\Ccal$
\end{algorithmic}
\end{algorithm}

The key part of this algorithm is searching for the largest biclique in step 3, which is known to be NP-hard \citep{Peeters} and computationally intractable.
Hence, we instead search for not a maximum but a maximal biclique, i.e., a biclique that no other biclique exactly contains.
We can enumerate maximal bicliques efficiently by an existing algorithm Bimax \citep{Prelic}.

%%%%%%%%%%%%%%%%%%%
\section{Improving the power of biclique tests}
\label{proposed method}

In this section, we propose a method to improve the power of the biclique test of \cite{Puelz}.
As we saw in the previous section, it is necessary to obtain a desirable biclique decomposition to attain high power in the biclique test.
\cite{Puelz} evaluated the relationship between biclique size and power and proposed the decomposition algorithm based on biclique size (Algorithm \ref{biclique decomposition}).
However, their evaluation ignores the information on the pattern of assignments, which may lead to selecting bicliques with large size but low power.
For example, suppose that when we test a contrast hypothesis $H^{\{\abf, \bbf\}}_0$, we obtained a biclique $C = (U, \Zcal)$ with the pattern shown in Figure \ref{undesirable biclique} (represented in a matrix form equivalent to a bipartite graph).
In this extreme case, for all the assignments $z \in \Zcal$, the treatment exposures to each unit $i \in U$ are completely the same.
Thus, the distribution of the test statistic degenerates to a single point and gains no power at all, no matter how large the size of the biclique.
This implies that not only the size of the biclique but also its pattern affects the power.
Therefore, we can expect to improve the power of the biclique test by evaluating the power of the randomization test more precisely and constructing a more desirable biclique decomposition based on the evaluation.

\begin{figure}[htbp]
	\begin{center}
		\includegraphics[width=5cm]{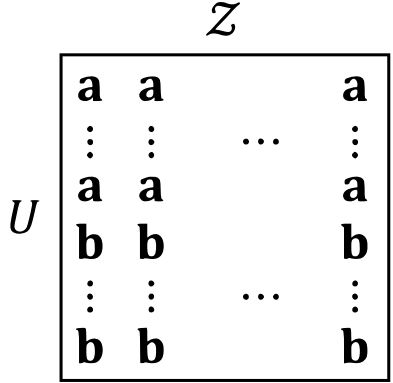}
		\caption{Example of an undesirable biclique.}
		\label{undesirable biclique}
	\end{center}
\end{figure}

In this paper, we will only focus on the contrast hypothesis.
In this case, evaluating the power of the randomization test in each biclique comes down to evaluating the power of the classical randomization test for the null hypothesis \eqref{eq: H_0}.

Since the power of the randomization test depends on the potential outcomes of all units, there are only a few results on it.
\cite{Krieger} evaluated the power of the randomization test when all the assignments are balanced (i.e., the size of the treatment and control groups are equal).
The results suggest that not only the size of the set of units and assignments but also the orthogonality of the assignment vectors affect the power.
We extend the argument of \cite{Krieger} to more general cases where the assignments are unbalanced and propose a biclique decomposition algorithm based on the evaluation.

\subsection{Power analysis of randomization tests}
\subsubsection{Derivation of  the power of randomization tests}
In this section, we denote the elements of the assignment vector $z$ by $1$ and $-1$ for convenience ($1$ : treatment, $-1$ : control).
We show the setup for the evaluation of the power in the following.
The experimental design $P(z)$ is uniform distribution over a given set of assignments $\Zbb = \{z^{(1)}, \dots, z^{(m)}\}$, i.e., $P(z^{(k)}) = 1/m \ (k = 1, \dots, m)$.
We then assume the following model for the potential outcomes:
\begin{align}
	Y(z) = \frac{\tau}{2} z + Y_0 \quad (\tau \in \Rbb, \ Y_0 \in \Rbb^N).
	\label{eq: model}
\end{align}
This is a model with a common treatment effect $\tau$ for all units, where the null hypothesis corresponds to $\tau=0$.
We refer to $Y_0$ as {\it base outcomes}.
Note that \cite{Krieger} discussed the model with covariates, but for the sake of simplicity, we will consider this simplest model without covariates.
The test statistic is the difference in means, 
\begin{align}
	T(z, Y) &= \ave\{Y_i \ | \ i \in \Ubb, \ z_i = 1\}
		- \ave\{Y_i \ | \ i \in \Ubb, \ z_i = -1\}. %\\
		%&= \frac{1}{|\{i | z_i = 1\}|} \sum_{i : z_i=1} Y_i - \frac{1}{|\{i | z_i = -1\}|} \sum_{i : z_i=-1} Y_i.
	\label{eq: ave diff}
\end{align}

The power of the randomization test with significance level $\alpha$ is expressed by the following formula:
\begin{align}
	\beta(\tau; \Zbb, Y_0)
		= \frac{1}{m} \sum_{k=1}^{m} \ind \left\{T(z^{(k)}, Y(z^{(k)})) > Q_{1-\alpha} \left(\{T(z^{(l)}, Y(z^{(k)}))\}_{l=1, \dots, m} \right) \right\}.
\end{align}
Here, $Q_{1-\alpha}(A)$ represents the lower $1-\alpha$ quantile of the set $A$, which corresponds to the $\lceil m(1-\alpha) \rceil$th number in the ascending order of $m$ elements in this case.
Thus, the power of the randomization test depends not only on the given assignment set $\Zbb$ but also on the unobserved base outcomes $Y_0$.
To make it possible to evaluate the power, we make the following assumption on $Y_0$.
\begin{ass}
\label{ass1}
The base outcome for each unit follows the same normal distribution independently:
\[
	Y_{0, i} \sim N(\mu, \sigma^2) \ i.i.d. \quad (i \in \Ubb).
\]
\end{ass}
Under this assumption, we will evaluate the average power
\begin{align}
\bar{\beta}(\tau; \Zbb) &= \Exp_{Y_0}[\beta(\tau; \Zbb, Y_0)] \\
	&= \frac{1}{m} \sum_{k=1}^{m} 
		P_{Y_0}\left(T(z^{(k)}, Y(z^{(k)})) > Q_{1-\alpha} \left(\{T(z^{(l)}, Y(z^{(k)}))\}_{l=1, \dots, m} \right) \right).
\end{align}

To make the problem simpler, we impose a further assumption on the pattern of assignments. %new
Now, for an assignment vector $z$, we define the transformed vector $\wtld{z}$ by
\[
	\wtld{z_i} = 
	\begin{dcases}
		\frac{1}{|\{j \in \Ubb| z_j = 1\}|} & (z_i = 1) \\
		-\frac{1}{|\{j \in \Ubb| z_j = -1\}|} & (z_i = -1)
	\end{dcases}
	\quad (i \in \Ubb).
\]

\begin{ass}
\label{ass2}
\begin{description}
\item[(a)] The proportion of treated units equals to some constant $p$ for all assignments in $\Zbb$:
\[
	\frac{|\{i \in \Ubb| z^{(k)}_{i} = 1\}|}{N} = p \in (0, 1) \quad (k = 1, \dots, m).
\]
\item[(b)] The inner product of the transformed assignment vector and the assignment vector equals to non-negative constant $\rho$ for all different pairs of assignments in $\Zbb$: 
\[
	\frac{\wtld{z^{(k)}} \cdot z^{(l)}}{2} = \rho \in [0, 1) \quad (k \neq l).
\]
\end{description}
\end{ass}

Here, $\rho$ is assumed to be non-negative to derive the power evaluation.
We do not consider the case $\rho = 1$, which corresponds to the situation where all the assignment vectors are the same and gain no power since the test statistic degenerates to a single point.
Note that \cite{Krieger} corresponds to the case where $p=1/2$.

The above two assumptions are a bit unrealistic, but we will discuss their validity in subsection \ref{validity}.
Under the assumptions, we get the following power evaluation.

\begin{thm}[power of randomization tests]
\label{power analysis}
We assume the model \eqref{eq: model}. Under Assumption \ref{ass1} and \ref{ass2}, the average power of the randomization test with significant level $\alpha$ using \eqref{eq: ave diff} as a test statistic is
\begin{align}
	\bar{\beta}(\tau; \Zbb)
		= \int \Fbin \left( \lfloor m\alpha \rfloor-1; \ m-1, \ \Phi(z-\Theta) \right) \phi(z) dz,
	\label{eq: power function}
\end{align}
where
\[
	\Theta = \frac{\tau}{\sigma} \sqrt{N \cdot p(1-p) \cdot (1-\rho)}, 
\]
$\Fbin(k; n, p)$ is the cumulative distribution function of a binomial distribution, and $\Phi(z)$ and $\phi(z)$ are the cumulative distribution function and probability density function of the standard normal distribution, respectively.
\end{thm}

%This expression is not analytically, but we can get an approximation by numerical integration.

\subsubsection{Properties of the power evaluation formula}
According to the obtained power evaluation \eqref{eq: power function}, the number of assignments $m$ and other parameters affect the power in different ways.
While the power depends on $m$ directly, it depends on the other parameters ($N, p, \rho, \tau, \sigma$) only through the quantity $\Theta$.
In the following, we will look at the properties of the formula.
Now, we regard the right-hand side of \eqref{eq: power function} as a function of $(\Theta, m, \alpha) \in \Rbb \times \Nbb \times (0, 1)$, and denote it as $\Pcal(\Theta, m; \alpha)$.

\begin{prop}
\label{prop alpha}
Given $\Theta = 0$, $\Pcal(\Theta, m; \alpha)$ is less than or equal to $\alpha$.
In particular, $\Pcal(\Theta, m; \alpha) = \alpha$ holds when $\lfloor m\alpha \rfloor = m\alpha$.
\end{prop}

This implies that the probability of type I error does not exceed the significance level $\alpha$ under the null hypothesis $\tau=0$, which is one of the proofs that $\Pcal(\Theta, m; \alpha)$ is reasonable as a power evaluation formula.
Then, the following property holds for the dependence of $\Theta$.

\begin{prop}
\label{prop Theta}
$\Pcal(\Theta, m; \alpha)$ is strictly increasing for $\Theta$.
\end{prop}

This means that we can get high power when parameters $(N, p, \rho, \tau, \sigma$) take values that increase $\Theta$.
The power is higher when the number of the units $N$ is large, the balance of the assignments $p$ is close to $1/2$, the inner product of the assignments $\rho$ is close to 0, the treatment effect $\tau$ is large, and the variance $\sigma^2$ of the base outcome of the units is small.
Also, we have the following property on the dependence of $m$.

\begin{prop}
\label{prop m}
Suppose that $\Theta>0$.
When $m$ is an integer satisfying $\lfloor m\alpha \rfloor = m\alpha$, $\Pcal(\Theta, m; \alpha)$ is strictly increasing for such $m$.
\end{prop}

That is, the larger the number of assignments $m$ is, the higher the power tends to be.
Together with Proposition \ref{prop Theta}, this is consistent with the statement by \cite{Puelz} that the power increases as the size of the biclique increases.

Figure $\ref{power curve}$ shows the graph of $\Pcal(\Theta, m; \alpha=0.05)$.
As shown in Proposition \ref{prop Theta} and \ref{prop m}, $\Pcal$ is monotonically increasing for $\Theta$ and $m$.
Note that the change of $\Theta$ affects the power significantly, while the impact of $m$ is relatively slight.
Hence,  for a given assignment set $\Zbb$, we can state that $\Theta$ is an important quantity that characterizes the power of the randomization test.

\begin{figure}[htbp]
	\begin{center}
		\includegraphics[width=10cm]{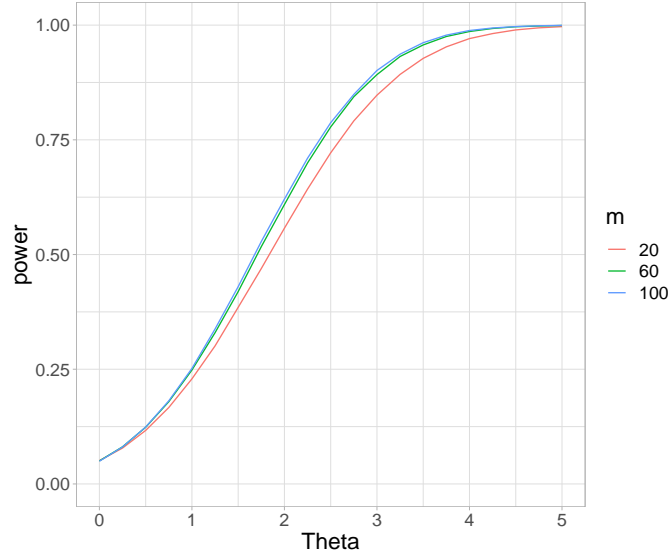}
		\caption{Graph of the power evaluation formula $\Pcal(\Theta, m; \alpha=0.05)$.}
		\label{power curve}
	\end{center}
\end{figure}

\subsubsection{Validity of the assumptions}
\label{validity}
We have derived the evaluation formula \eqref{eq: power function} under Assumption \ref{ass1} and \ref{ass2}.
In this subsection, we will examine the validity of these two assumptions.

First, we consider Assumption \ref{ass1}.
Assumption \ref{ass1} enables us to derive the power evaluation \eqref{eq: power function} and it is difficult to get the concrete expression for the power without normality on the base outcomes.
However, the result equivalent to Assumption \ref{ass1} is approximately justified when $N$ is sufficiently large, even if the base outcomes do not follow a normal distribution.

The derivation of \eqref{eq: power function} is based on the fact that $\{\wtld{z^{(k)}} \cdot Y_0\}_{k=1, \dots, m}$ follows multivariate normal distribution:
\begin{align}
	\sqrt{N} \left(
	\begin{array}{c}
		\wtld{z^{(1)}} \cdot Y_0 \\
		\vdots \\
		\wtld{z^{(m)}} \cdot Y_0
	\end{array}
	\right)
	\sim N_m (\bm{0}, \Sigma),
\end{align}
where we denote
\[
	\Sigma = \frac{\sigma^2}{p(1-p)}\left(
	\begin{array}{cccc}
		1 & \rho & \cdots & \rho \\
		\rho & 1 & \ddots & \vdots \\
		\vdots & \ddots & \ddots & \rho \\
		\rho & \cdots & \rho & 1 
	\end{array}	
	\right).
\]
This result holds asymptotically by Lindeberg-Feller central limit theorem without normality on $Y_0$.

\begin{prop}
\label{clt}
Suppose that base outcomes for each unit follow the same distribution with finite second moment independently, i.e. $Y_{0, i} \sim F \ i.i.d., \ \Exp[Y_{0, i}] = \mu, \ \Var[Y_{0,i}] = \sigma^2 \ (i = 1, \dots, N)$.
Then,
\begin{align}
	\sqrt{N} \left(
	\begin{array}{c}
		\wtld{z^{(1)}} \cdot Y_0 \\
		\vdots \\
		\wtld{z^{(m)}} \cdot Y_0
	\end{array}
	\right)
	\xrightarrow{d} N_m (\bm{0}, \Sigma) \quad (N \rightarrow \infty)
\end{align}
holds under Assumption \ref{ass2}.
\end{prop}

Next, we examine Assumption \ref{ass2}.
Assumption \ref{ass2} is hard for a given assignment set to satisfiy except for a few examples.
However, as we will see through the following numerical experiments, the evaluation \eqref{eq: power function} holds approximately even for assignment sets $\Zbb$ not satisfying Assumption \ref{ass2} by replacing $\Theta$ with a natural alternative quantity
\begin{align}
	\hat{\Theta} =
	\frac{\tau}{\sigma}\sqrt{N \cdot \hat{p}(1-\hat{p}) \cdot (1-\hat{\rho})}.
	\label{eq: Theta_hat}
\end{align}
Here, $\hat{p}$ and $\hat{\rho}$ denote
\begin{align}
	\hat{p} &= \frac{1}{m} \sum_{k=1}^{m} \frac{|\{i \in \Ubb| z^{(k)}_{i} = 1\}|}{N}, \\
	\hat{\rho} &= \max \left\{ \frac{1}{m(m-1)} \sum_{k=1}^{m} \sum_{l \neq k} 
		\frac{\wtld{z^{(l)}} \cdot z^{(k)}}{2}, \ 0 \right\} \\
\end{align}
respectively.

We perform the experiment with the following setup.
First, we generate assignment sets $\Zbb = \{z^{(1)}, \dots, z^{(m)}\}$ according to the following procedure:
\begin{align}
	z^{(k)}_i = 2B_{i, k} -1, \ B_{i, k} \sim Bern(p) \ i.i.d \quad (i = 1, \dots, N, \ k = 1, \dots, m).
	\label{eq: generate assignment}
\end{align}
However, since the value of the test statistic \eqref{eq: ave diff} is not defined for the assignment vectors whose elements all take the same value, we exclude them appropriately.
The assignment sets $\Zbb$ generated by this procedure do not satisfy Assumption \ref{ass2} in most cases.
For these $\Zbb$, we compare the actual average power $\bar{\beta}(\tau; \Zbb)$ with the estimated average power $\Pcal(\hat{\Theta}, m; \alpha)$.

We set two distributions for the base outcomes, $N(0, 1)$ and $Ex(1)$, as the cases where Assumption \ref{ass1} is satisfied and not satisfied, respectively.
We fix the treatment effect and significance level at $\tau=0.5$ and $\alpha=0.05$.
We conducted the experiment for 90 assignments generated under different combinations of $N\in\{20, 40, \dots, 300\}, m\in\{20, 60, 100\}$ and $p\in\{0.2, 0.5\}$. 
The actual average power $\bar{\beta}(\tau; \Zbb)$ is estimated by Monte Carlo approximation with 100 samples.

The results are shown in Figure \ref{simulation}.
The upper figure and lower one correspond to the case where base outcomes follow $N(0, 1)$ and $Ex(1)$ respectively.
Each plot corresponds to a randomly generated assignment set, where the horizontal axis and vertical axis represent $\hat{\Theta}$ and the actual power $\bar{\beta}(\tau; \Zbb)$ respectively.
The black curve represents the power evaluation formula $\Pcal(\hat{\Theta}, m; \alpha)$, on which each plot lies exactly when the assignment set satisfies Assumption \ref{ass1} and \ref{ass2}.

The upper figure corresponds to the situation where Assumption \ref{ass1} is satisfied, but Assumption \ref{ass2} is not.
Although each plot is slightly out of the curve, the evaluation formula $\Pcal(\hat{\Theta}, m; \alpha)$ explains the actual power $\bar{\beta}(\tau; \Zbb)$ very well.
On the other hand, in the lower figure, where Assumption \ref{ass1} are also not satisfied, the gap between the plots and curve gets larger, but the curve still roughly captures the trend of the plots.
Note that the gap almost disappears when $N$ is large, which explains the result of Proposition \ref{clt} that the evaluation formula \eqref{eq: power function} is justified in large samples even if normality does not hold.

From the above discussion, we can conclude that the power evaluation $\Pcal(\hat{\Theta}, m; \alpha)$ is still reasonable even when Assumption \ref{ass1} and \ref{ass2} are not satisfied, and that $\hat{\Theta}$, an alternative to $\Theta$, is the important quantity characterizing the power of randomization tests.

\begin{figure}[htbp]
	\begin{minipage}{1\hsize}
	  	\begin{center}
		   \includegraphics[width=14cm]{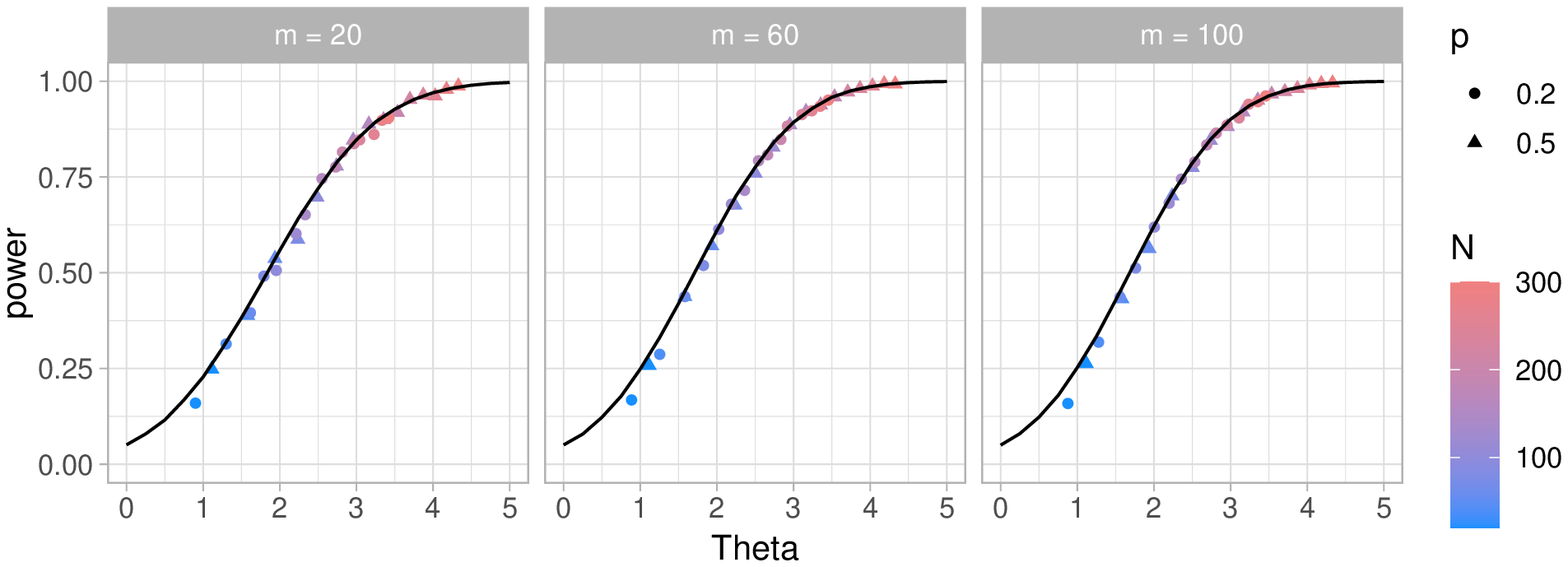}
		\end{center}
	\end{minipage}
	\begin{minipage}{1\hsize}
		\begin{center}
			\includegraphics[width=14cm]{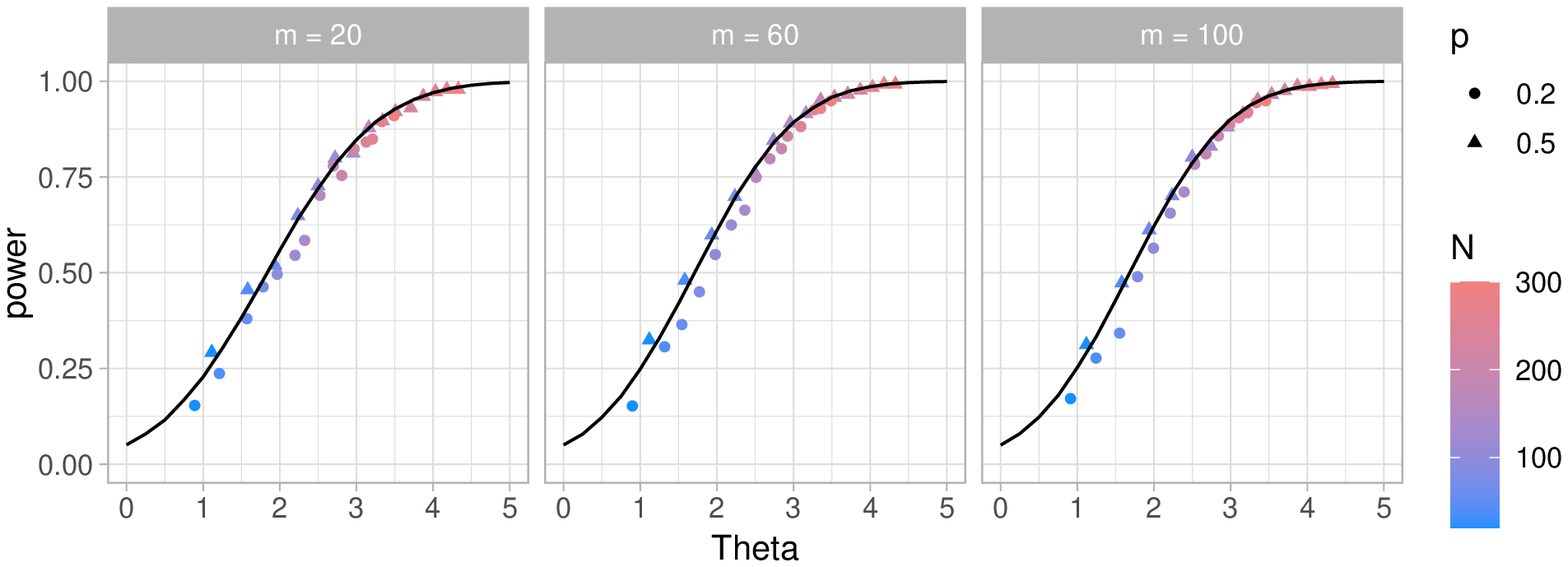}
		\end{center}
	\end{minipage}
	\caption{Comparison between the actual average power $\bar{\beta}(\tau; \Zbb)$ (plot) and the power evaluation formula $\Pcal(\hat{\Theta}, m; \alpha)$ (curve).
	The upper figure and lower one correspond to the case where base outcomes follow $N(0, 1)$ and $Ex(1)$ respectively.}
	\label{simulation}
\end{figure}

\subsection{Modified biclique decomposition algorithm}
From the discussion in the previous section, we found that $\hat{\Theta}$ characterizes the power for a given assignment set $\Zbb$.
In particular, the quantity $\hat{\Theta}_0 =\sqrt{N \cdot \hat{p}(1-\hat{p}) \cdot (1-\hat{\rho})}$, which excludes from $\hat{\Theta}$ the terms unrelated to $\Zbb$, the treatment effect $\tau$ and the variance of base outcomes $\sigma^2$, is an index expressing the ``goodness'' of $\Zbb$.

In the biclique test for the contrast hypothesis $H^{\{\abf, \bbf\}}$, we run the classical randomization test in each biclique $C$.
Hence, we can conclude that $\hat{\Theta}_0$, which is calculated from $C$, is an important quantity representing the desirability of the biclique.
In this case, $\hat{\Theta}_0$ is calculated as follows.
For example, for a biclique $C = (U, \Zcal) \ (|U| = 3, |\Zcal| = 4)$ with pattern
\[
	\left(
	\begin{array}{cccc}
		\abf & \bbf & \bbf & \abf \\
		\abf & \bbf & \abf & \bbf \\
		\bbf & \abf & \bbf & \bbf \\
	\end{array}	
	\right),
\]
we make $\abf$ and $\bbf$ correspond to $1$ and $-1$ respectively, and regard that
\begin{align}
	z^{(1)} &= (1, 1, -1)^{\T}, \\
	z^{(2)} &= (-1, -1, 1)^{\T}, \\
	z^{(3)} &= (-1, 1, -1)^{\T}, \\
	z^{(4)} &= (1, -1, -1)^{\T}.
\end{align}
Then, we calculate $\hat{\Theta}_0$ according to \eqref{eq: Theta_hat}.
In this example, we can calculate that $\hat{p} = 5/12, \ \hat{\rho} = \max \{-1/4, 0\} = 0$ and $\hat{\Theta}_0 = \sqrt{3 \cdot 5/12(1-5/12) \cdot (1-0)} \fallingdotseq 0.85$.

Thus, we can expect that the power of the biclique test can be improved by modifying Algorithm \ref{biclique decomposition} to select bicliques according to $\hat{\Theta}_0$.
The modified algorithm is shown in Algorithm \ref{modified decomposition}.
It is difficult to implement step 3 as with Algorithm \ref{biclique decomposition}.
Therefore, in practice, we enumerate maximal bicliques and choose the biclique with the largest $\hat{\Theta}_0$ among them.

This modification reduces the risk of selecting bicliques with low power despite their large size.
For example, the undesirable biclique shown in Figure \ref{undesirable biclique} is unlikely to be selected by Algorithm 2 since $\hat{\Theta}_0 = 0$ followed by $\hat{\rho} = 0$.

\begin{algorithm}[H]
\caption{Modified biclique decomposition algorithm}
\label{modified decomposition}
\begin{algorithmic}[1]
	\Require null exposrure graph $G^{\Fcal}_{f}$
	\Ensure biclique decomposition $\Ccal$
	\State $\Ccal \leftarrow \emptyset, \ G \leftarrow G^{\Fcal}_{f}$ 
	\While{$|\Zcal(G)| > 0$}
		\State $C^{\ast} = \mathrm{argmax}_{C \in G} \hat{\Theta}_0(C)$
		\State $E(G) \leftarrow E(G) \setminus E(C^{\ast}), \ 
			\Zcal(G) \leftarrow \Zcal(G) \setminus \Zcal(C^{\ast})$
		\State $\Ccal \leftarrow \Ccal \cup \{C^{\ast}\}$
	\EndWhile
	\State \Return $\Ccal$
\end{algorithmic}
\end{algorithm}

%%%%%%%%%%%%%%%%%%%
\section{Simulation}
\label{numerical experiment}

\subsection{Comparison of biclique decompositions}
\label{simulation1}
The proposed method (Algorithm \ref{modified decomposition}) should select bicliques with a larger value of $\hat{\Theta}_0$ than the existing method (Algorithm \ref{biclique decomposition}).
In this section, we will see how the difference between the two methods depends on the structure of the null exposure graph through a numerical experiment.

As a virtual null exposure graph of some contrast hypothesis $H^{\{\abf, \bbf\}}_0$, we randomly generate a bipartite graph with density $d$ and balance $b$.
Then, we construct biclique decompositions by Algorithm \ref{biclique decomposition} and \ref{modified decomposition}, and compare their features.
Here, the density of the graph is the ratio of the actual number of edges to the total number of possible edges, and the balance is the ratio of the number of edges corresponding to treatment exposure $\abf$ to that of all edges.
The size of the generated null exposure graph is $N = m = 1000$, and we conduct simulations under different parameter settings of $d\in\{0.8, 0.9\}$ and $b\in\{0.01, 0.1, 0.5\}$.

Note that, for Algorithm \ref{biclique decomposition} and \ref{modified decomposition}, searching for the biclique with the largest $|E(C)|$ or $\hat{\Theta}(C)$ in step 3 is computationally intractable.
Instead, we will approximate it by enumerating 10000 maximal bicliques by Bimax and selecting the biclique with the largest $|E(C)|$ or $\hat{\Theta}(C)$ among them.
Bimax allows the user to set the minimum size of the biclique to be enumerated, which in this case we set to be $|U|, |\Zcal| \geq 20$.
For this implementation, including the experiments in the next section, we used the R package CliqueRT \citep{CliqueRT}.

The results are shown in Figure \ref{comparing_decomposition}.
The vertical and horizontal axis represents the size of a biclique $|E(C)|$ and $\hat{\Theta}_0(C)$ respectively, and each plot corresponds to each biclique in the obtained biclique decompositions.
Here, a few outliers are excluded.
When $b=0.01, 0.1$, where the balance is unbalanced, the proposed method selects bicliques with a larger value of $\hat{\Theta}_0$ than the existing method, and the difference between the methods is especially significant when the unbalance is large.
On the other hand, when $b=0.5$, where the balance is even, the plots are located at almost the same place and there is no clear difference between the methods.
In summary, the difference between the biclique decompositions constructed by the proposed method and the existing method depends on the structure of the null exposure graph, and the difference is especially significant when the imbalance of the graph is large.

\begin{figure}[htbp]
	\begin{center}
		\includegraphics[width=14cm]{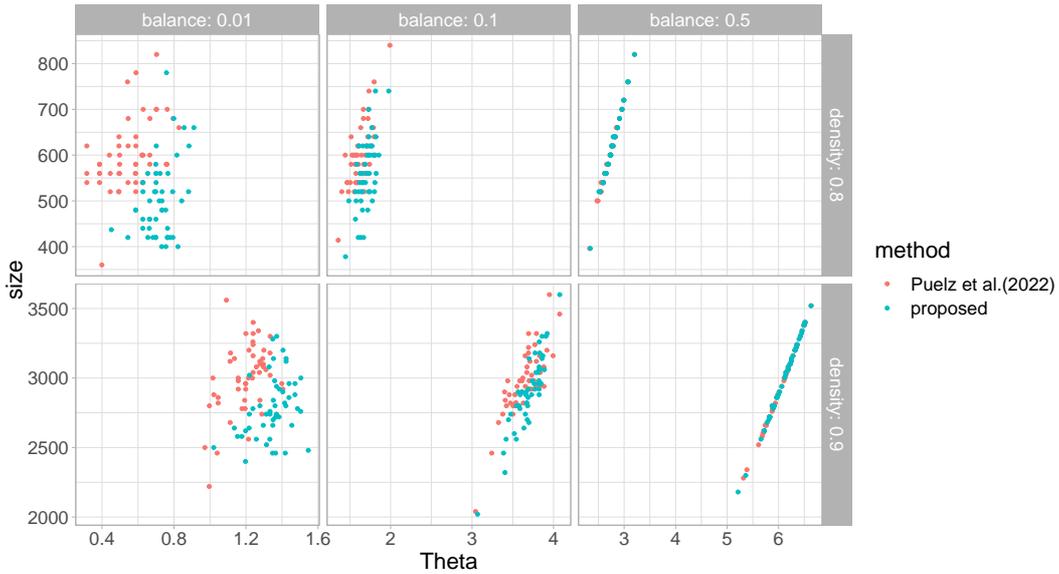}
		\caption{Comparison of biclique decompositions obtained by the existing method and proposed one.
The plots correspond to each biclique in the biclique decomposition.}
		\label{comparing_decomposition}
	\end{center}
\end{figure}

\subsection{Comparison of power}
We compare the power of the biclique test based on biclique decomposition constructed by the existing method (Algorithm \ref{biclique decomposition}) and the proposed method (Algorithm \ref{modified decomposition}).

The setup follows Example \ref{spatial}, a scenario where spatial interference exists.
Let $x_i \in \Rbb^2$ be the coordinates of each unit and let the exposure mapping function be
\begin{align}
	f_i(z) = 
	\begin{cases}
		2 & (z_i = 1) \\
		1 & (z_i = 0 \text{ and } z_j = 1, ||x_i - x_j|| \leq r \text{ for some $j \in \Ubb$}) \\
		0 & (\text{otherwise}).
	\end{cases}
\end{align}
Now, we want to test the null hypothesis $H^{\{0, 1\}}_0$.
We set the number of units to $N=1000$ and place each unit randomly according to the following procedure:
\begin{align}
	x_i \sim
	\begin{cases}
		N\left((0.5, 0.5)^{\T}, 0.1^2 I_2)\right) & (i = 1, \dots, 500) \\
		N\left((0.25, 0.75)^{\T}, 0.075^2 I_2)\right) & (i = 501, \dots, 800) \\
		N\left((0.3, 0.3)^{\T}, 0.075^2 I_2)\right) & (i = 801, \dots, 1000).
	\end{cases}
\end{align}
We generate $m=1000$ assignments according to the same procedure as \eqref{eq: generate assignment}, where the parameter is set to $p\in\{0.1, 0.2\}$.
We set the distance of interference to $r\in\{0.005, 0.01, 0.05\}$ and the significance level of the test to $\alpha=0.05$.
The potential outcome of each unit is generated according to the following procedure:
\begin{align}
	 Y_i(0) \sim N(0, 1) \ i.i.d., \quad Y_i(1) = Y_i(0) + \tau \quad (i = 1, \dots, N),
\end{align}
where the potential outcome under the treatment exposure $\abf\in\{0, 1\}$ is simply written as $Y_i(\abf)$. 
We use the difference in means as the test statistic, but if there is an assignment that give the same treatment exposure to all units, the difference in means is not defined.
Hence, we use the following modified test statistic for convenience:
\begin{align}
	T(z, Y; C) = 
	\begin{cases}
		0 \quad (f_i(z) = \abf \text{ for all $i \in U$} \ (\abf \in \{0, 1\})) \\
		\ave\{Y_i \ | \ i \in U, \ f_i(z) = 1\} - \ave\{Y_i \ | \ i \in U, \ f_i(z) = 0\} \quad (\text{otherwise}).
	\end{cases}
\end{align}

Under the above setup, we look at the average power for different $\tau$ ($\tau=0$ corresponds to the null hypothesis).
The power is estimated by Monte Carlo approximation with 100 samples.

The results are shown in Figure \ref{spatial experiment}.
We can see that the proposed method outperforms the existing method in all settings.
Note that the difference in power depends on the parameters $(p, r)$, which is due to the difference in the structure of the null exposure graph.
Table \ref{NEgraph data} shows the characteristic values (density and balance) of the null exposure graphs obtained under each parameter setting.
First, the more sparse the graph is, the greater the improvement of the proposed method on the existing method.
Second, the more unbalanced the graph is, the greater the power improvement, which is consistent with the result of section \ref{simulation1} that the difference between the two methods is significant when the imbalance of the graph is large.
In summary, when the null exposure graph is sparse and unbalanced, the power gain of the proposed method over the existing method is large.

Both density and balance of null exposure graphs are important factors that affect the power of the biclique test.
The higher the density of the graph, the larger the size of obtained bicliques, which results in the higher power of the biclique test \citep{Puelz}.
Also, when the balance of the graph is close to even, the balance of each biclique is close to even, i.e., $\hat{p}$ takes a value close to $1/2$, which leads to high power too.
In such a situation where the null exposure graph is dense or the balance is even, we can obtain sufficiently desirable bicliques either by the existing method or by the proposed one, and the improvement of the power by the proposed method is small.
On the other hand, the situation where the null exposure graph is sparse and highly unbalanced is a particularly difficult setting, which will lead to a significant loss of power if we do not select bicliques carefully.
In such a case, the existing method considering only the size of the biclique greatly loses the power, while the proposed method considering other factors such as balance and orthogonality maintains high power.
In summary, the proposed method shows higher power than the existing method, and it is especially effective in difficult situations for the existing method.

\begin{figure}[htbp]
	\begin{center}
		\includegraphics[width=14cm]{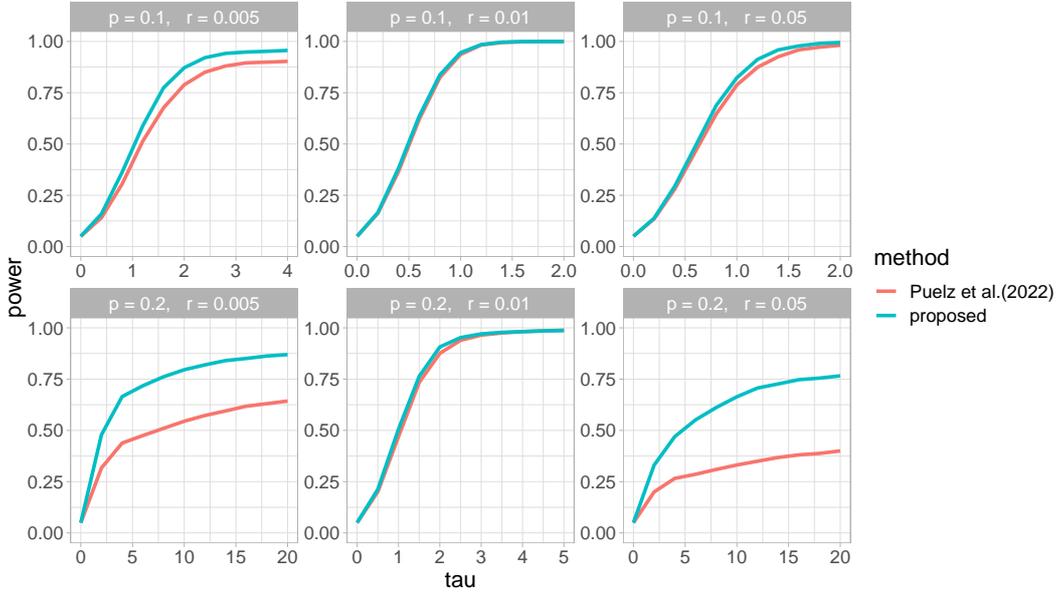}
		\caption{Comparison of power between the proposed method and the existing method in simulations of spatial interference setting.}
		\label{spatial experiment}
	\end{center}
\end{figure}

\begin{table}
	\caption{Characteristic values of the null exposure graphs obtained under each parameter setting.}
	\centering
	\label{NEgraph data}
	\fbox{%
	\begin{tabular}{c||ccc|ccc}
		$p$			& \multicolumn{3}{c|}{0.1} & \multicolumn{3}{c}{0.2} \\ \hline
		$r$        	& 0.005 & 0.01 & 0.05 & 0.005 & 0.01 & 0.005 \\ \hline \hline
		density    	& 0.900 & 0.900 & 0.900 & 0.799 & 0.800 & 0.799 \\
		balance   	& 0.030 & 0.129 & 0.877 & 0.055 & 0.228 & 0.948 \\
	\end{tabular}}
\end{table}

%%%%%%%%%%%%%%%%%%%
\section{Conclusion}
\label{conclusion}
In this paper, we proposed a method to improve the power of the biclique test \citep{Puelz} based on the power evaluation of the randomization test.
One of our main contributions is deriving a concrete expression for the power of the randomization test under several assumptions and examining its properties.
According to the derived formula, the power of the randomization test is characterized by the number of assignments $m$ and the quantity $\Theta$ that depends on the number of units, the balance of assignments, the orthogonality of assignments, the treatment effect, and the variance of outcomes.
Based on this fact, we have improved the power of the biclique test by modifying the biclique decomposition procedure to obtain more desirable bicliques.
Through simulations in a spatial interference setting, we confirmed that the proposed method shows higher power than the existing method.

There are several possible directions for future works.
First, in this paper, we derived the power of the randomization test only for the contrast hypothesis with the difference in means type test statistic.
Hence, it is not clear whether the proposed method is also effective when other test statistics are used.
In addition, the proposed method cannot be applied for any other type of null hypothesis than the contrast hypothesis.
In this sense, the scope of the proposed method is limited.
Therefore, to improve the power of the biclique test in other settings, it is necessary to investigate how to obtain a refined biclique decomposition under individual or more general settings.

In addition, the proposed method does not use any information on covariates.
However, if covariate information is available, we can use it to improve the efficiency of the inference \citep{Morgan_and_Rubin, Hennessy_etal}.
For example, it may be possible to improve the power of the biclique test by evaluating the power incorporating covariates and selecting bicliques taking the information of the covariates into account.

\appendix

\section{Proofs}

\subsection{Proofs of theorems and propositions}

We give proofs of all theorems and propositions except Theorem 1, 2, and 3.
See \cite{Basse} for the proofs of Theorem 1 and 2, and see \cite{Puelz} for that of Theorem 3.

\begin{proof}[Proof of Theorem~\ref{power analysis}]
Since the assumptions make each assignment symmetric, we can express the average power as
\begin{align}
	\bar{\beta}(\tau; \Zbb) &= \frac{1}{m} \sum_{k=1}^{m} 
		P_{Y_0}\left(T(z^{(k)}, Y(z^{(k)})) > Q_{1-\alpha} \left(\{T(z^{(l)}, Y(z^{(k)}))\}_{l=1, \dots, m} \right) \right) \\
		&= P_{Y_0}\left(T(z^{(1)}, Y(z^{(1)})) > Q_{1-\alpha} \left(\{T(z^{(k)}, Y(z^{(1)}))\}_{k=1, \dots, m} \right) \right).
\end{align}

Now, a simple calculation shows
\begin{align}
	\wtld{z^{(k)}} \cdot z^{(l)} = 
	\begin{dcases}
		2 & (k=l) \\
		2\rho & (k \neq l)
	\end{dcases}, \quad
	\wtld{z^{(k)}} \cdot \wtld{z^{(l)}} =
	\begin{dcases}
		\frac{1}{Np(1-p)} & (k=l) \\
		\frac{\rho}{Np(1-p)} & (k \neq l)
	\end{dcases}
	\label{eq: inner product}
\end{align}
under Assumption \ref{ass2} (Lemma \ref{lem inner product}).
Considering that the test statistic is expressed as
\begin{align}
	T(z^{(1)}, Y(z^{(1)})) &= \wtld{z^{(1)}} \cdot Y(z^{(1)}) =  \tau + \wtld{z^{(1)}} \cdot Y_0, \\
	T(z^{(k)}, Y(z^{(1)})) &= \wtld{z^{(k)}} \cdot Y(z^{(1)}) = \rho \tau + \wtld{z^{(k)}} \cdot Y_0 \quad (k = 2, \dots, m)
\end{align}
and the moments of normally distributed terms are
\begin{align}
\begin{aligned}
	\Exp[\wtld{z^{(k)}} \cdot Y_0]
		&= 0 \quad (k = 1, \dots, m), \\
	\Var[\wtld{z^{(k)}} \cdot Y_0]
		&= \| \wtld{z^{(k)}} \|^{2} \sigma^2 = \frac{\sigma^2}{Np(1-p)} \quad (k = 1, \dots, m), \\
	\Cov[\wtld{z^{(k)}} \cdot Y_0, \ \wtld{z^{(l)}} \cdot Y_0]
		&= \left( \wtld{z^{(k)}} \cdot \wtld{z^{(l)}} \right) \sigma^2 = \frac{\rho \sigma^2}{Np(1-p)} \quad (k \neq l),
	\label{eq: moment2}
\end{aligned}
\end{align}
the average power reduces to the following equation:
\[
	\bar{\beta}(\tau; \Zbb) = P(V_1 > Q_{1-\alpha}(\{V_k\}_{k=1, \dots, m})).
\]
Here, $V_1, \dots, V_m$ are random variables represented as
\begin{align}
	V_1 &= \tau + \tilde{\sigma} \sqrt{\rho} Z_0 + \tilde{\sigma} \sqrt{1-\rho} Z_1,\\
	V_k &= \rho \tau + \tilde{\sigma} \sqrt{\rho} Z_0 + \tilde{\sigma} \sqrt{1-\rho} Z_k \quad (k = 2, \dots, m),
\end{align}
where $Z_0, Z_1, \dots, Z_m$ follows the standard normal distribution independently and 
\[
	\tilde{\sigma} = \frac{\sigma}{\sqrt{Np(1-p)}}.
\]
This is the probability that the number of $V_k \ (k = 2, \dots, m)$ greater than or equal to $V_1$ is less than or equal to $\lfloor m\alpha \rfloor -1$. 

Now, conditioned on the event $Z_1=z$, the events $V_k \geq V_1 \ (k = 2, \dots, m)$ are each equivalent to the event
\[
	Z_k \geq z + \frac{\tau}{\tilde{\sigma}} \sqrt{1-\rho}.
\]
Since these events are conditionally independent of each other, the conditional probability $P(V_1 > Q_{1-\alpha}(\{V_k\}_{k=1, \dots, m}) | Z_1 = z)$ is equal to the probability that the success count is less than or equal to $\lfloor m\alpha \rfloor -1$ in $m-1$ independent Bernoulli trials with success probability
\[
	P\left(Z_k \geq z + \frac{\tau}{\tilde{\sigma}} \sqrt{1-\rho} \right) = \Phi \left(-z - \frac{\tau}{\tilde{\sigma}} \sqrt{1-\rho} \right),
\]
which is expressed as $\Fbin(\lfloor m\alpha \rfloor-1; \ m-1, \ \Phi(-z-\tau \sqrt{1-\rho}/\tilde{\sigma}))$.
Therefore, the average power is given by
\begin{align}
	\bar{\beta}(\tau; \Zbb) &= \int P(V_1 > Q_{1-\alpha}(\{V_l\}_{l=1, \dots, m}) | Z_1 = z) \phi(z) dz \\
	&= \int \Fbin\left(\lfloor m\alpha \rfloor-1; \ m-1, \ \Phi \left(-z-\frac{\tau}{\tilde{\sigma}} \sqrt{1-\rho} \right)\right) \phi(z) dz \\
	&= \int \Fbin(\lfloor m\alpha \rfloor-1; \ m-1, \ \Phi(z - \Theta)) \phi(z) dz.
\end{align}
\end{proof}

%%%%%%%%%%%%%%%%%%%%%

%\begin{prop}
%\label{prop alpha}
%Given $\Theta = 0$, $\Pcal(\Theta, m; \alpha)$ is less than or equal to $\alpha$.
%In particular, $\Pcal(\Theta, m; \alpha) = \alpha$ holds when $\lfloor m\alpha \rfloor = m\alpha$.
%\end{prop}

\begin{proof}[Proof of Proposition~\ref{prop alpha}]
By simple calculation,
\begin{align}
	\Pcal(0, m; \alpha)
	&= \int \Fbin \left( \lfloor m\alpha \rfloor-1; \ m-1, \ \Phi(z) \right) \phi(z) dz \\
	&= \int^{1}_{0} \Fbin \left( \lfloor m\alpha \rfloor-1; \ m-1, \ t \right) dt \\
	&= \int^{1}_{0} \sum_{k=0}^{\lfloor m\alpha \rfloor - 1} \binom{m-1}{k} t^k (1-t)^{m-k-1} dt \\
	&= \sum_{k=0}^{\lfloor m\alpha \rfloor - 1} \binom{m-1}{k} B(k+1, m-k) \\
	&= \sum_{k=0}^{\lfloor m\alpha \rfloor - 1} \frac{(m-1)!}{k! (m-k-1)!}
\cdot \frac{k! (m-k-1)!}{m!} \\
	&= \sum_{k=0}^{\lfloor m\alpha \rfloor - 1} \frac{1}{m} \\
	&= \frac{\lfloor m\alpha \rfloor}{m} \leq \alpha.
\end{align}
holds.
The equality holds if and only if $\lfloor m \alpha \rfloor = m \alpha$.
\end{proof}

%\begin{prop}
%\label{prop Theta}
%$\Pcal(\Theta, m; \alpha)$ is strictly increasing for $\Theta$.
%\end{prop}

\begin{proof}[Proof of Proposition~\ref{prop Theta}]
For $\Theta_1 < \Theta_2$ and any $z \in \Rbb$, 
\begin{align}
	\Phi(z - \Theta_1) &> \Phi(z - \Theta_2), \\
	\text{hence} \quad \Fbin(\lfloor m\alpha \rfloor-1; \ m-1, \ \Phi(z - \Theta_1))
		&< \Fbin(\lfloor m\alpha \rfloor-1; \ m-1, \ \Phi(z - \Theta_2))
\end{align}
holds.
Thus, multiplying both sides by $\phi(z)$ and integrating over $z$ yields $\Pcal(\Theta_1, m; \alpha) < \Pcal(\Theta_2, m; \alpha)$.
\end{proof}

%\begin{prop}
%\label{prop m}
%Suppose that $\Theta>0$.
%When $m$ is an integer satisfying $\lfloor m\alpha \rfloor = m\alpha$, $\Pcal(\Theta, m; \alpha)$ is strictly increasing for such $m$.
%\end{prop}

\begin{proof}[Proof of Proposition~\ref{prop m}]
We will show $\Pcal(\Theta, m_1; \alpha) < \Pcal(\Theta, m_2; \alpha)$ for integers $m_1 < m_2$ satisfying $\lfloor m_i\alpha \rfloor = m_i\alpha \ (i=1, 2)$.
We calculate
\begin{align}
	\Pcal(\Theta, m; \alpha)
	&= \int \Fbin \left( \lfloor m\alpha \rfloor-1; \ m-1, \ \Phi(z-\Theta) \right) \phi(z) dz \\
	&= \int_{0}^{1} \Fbin \left( \lfloor m\alpha \rfloor-1; \ m-1, \ t \right) \frac{\phi(\Phi^{-1}(t) + \Theta)}{\phi(\Phi^{-1}(t))} dt
\end{align}
and define
\[
	g(t; \Theta) = \frac{\phi(\Phi^{-1}(t) + \Theta)}{\phi(\Phi^{-1}(t))}.
\]
Then, we can show that $g(t; \Theta)$ is strictly decreasing for $t$ (Lemma \ref{g property}).
Next, for
\begin{align}
	\Pcal&(\Theta, m_2; \alpha) - \Pcal(\Theta, m_1; \alpha) \\
	&= \int_{0}^{1} \Bigl\{\Fbin(\lfloor m_2\alpha \rfloor-1; \ m_2-1, \ t) - \Fbin (\lfloor m_1\alpha \rfloor-1; \ m_1-1, \ t) \Bigr\} g(t; \Theta) dt
\end{align}
we define
\[
	h(t) = \Fbin(\lfloor m_2\alpha \rfloor-1; \ m_2-1, \ t) - \Fbin (\lfloor m_1\alpha \rfloor-1; \ m_1-1, \ t).
\]
The function $h(t)$ satisfies the following two properties (Lemma \ref{h property} (\cite{Krieger})):
\begin{description}
	\item[$(\mathrm{i})$] There exists $t_0 \in (0, 1)$ such that $h(t)>0 \ (t<t_0), \ h(t)<0 \ (t>t_0)$.
	\item[$(\mathrm{ii})$] $\int_{0}^{1}h(t)dt = 0$.
\end{description}

Therefore, we have
\begin{align}
	\Pcal(\Theta, m_2; \alpha) - \Pcal(\Theta, m_1; \alpha)
	 &= \int_{0}^{t_0} g(t; \Theta) h(t) dt + \int_{t_0}^{1} g(t; \Theta) h(t) dt \\
	 &> g(t_0; \Theta) \int_{0}^{t_0} h(t) dt + g(t_0; \Theta) \int_{t_0}^{1} h(t) dt \\
	 &= g(t_0; \Theta) \int_{0}^{1} h(t) \\
	 &= 0.
\end{align}
\end{proof}

%\begin{prop}
%\label{clt}
%Suppose that base outcomes for each unit follow the same distribution with finite second moment independently, i.e. $Y_{0, i} \sim F \ i.i.d., \ \Exp[Y_{0, i}] = \mu, \ \Var[Y_{0,i}] = \sigma^2 \ (i = 1, \dots, N)$.
%Then,
%\begin{align}
%	\sqrt{N} \left(
%	\begin{array}{c}
%		\wtld{z^{(1)}} \cdot Y_0 \\
%		\vdots \\
%		\wtld{z^{(m)}} \cdot Y_0
%	\end{array}
%	\right)
%	\xrightarrow{d} N_m (\bm{0}, \Sigma) \quad (N \rightarrow \infty)
%\end{align}
%holds under Assumption \ref{ass2}.
%Here, we denote
%\[
%	\Sigma = \frac{\sigma^2}{p(1-p)}\left(
%	\begin{array}{cccc}
%		1 & \rho & \cdots & \rho \\
%		\rho & 1 & \ddots & \vdots \\
%		\vdots & \ddots & \ddots & \rho \\
%		\rho & \cdots & \rho & 1 
%	\end{array}	
%	\right).
%\]
%\end{prop}

\begin{proof}[Proof of Proposition~\ref{clt}]
We will show it by the Lindeberg-Feller theorem.
Define the matrix $A \in \Rbb^{m \times N}$ as $A = \sqrt{N}\left(\wtld{z^{(1)}}, \dots, \wtld{z^{(m)}}\right)^{\T}$, and let its column vectors be $a_1, \dots, a_N$.
The claim to be shown is equivalent to
\[
	AY_0 = \sum_{i=1}^{N} a_i Y_{0,i} \xrightarrow{d} N_{m}(\bm{0}, \Sigma) \quad (N \rightarrow \infty).
\]
The mean and variance of multivariate normal distribution to converge will be
\begin{align}
	\sum_{i=1}^{N} \Exp \left[a_i Y_{0,i}\right]
	&= \Exp \left[\sum_{i=1}^{N} a_i Y_{0,i}\right]
	= E[AY_0] = \mu A \ind = \bm{0}, \\
	\sum_{i=1}^{N} \Var \left[a_i Y_{0,i}\right] 
	&= \Var \left[\sum_{i=1}^{N} a_i Y_{0,i}\right] 
	= \Var[AY_0]
	= A \Var[Y_0] A^{\T} = \sigma^2 A A^{\T} = \Sigma.
\end{align}
respectively.
Now, it suffices to check Lindeberg's condition
\[
	\sum_{i=1}^{N} \Exp \left[||a_i Y_{0,i}||^{2} \ind\{||a_i Y_{0,i}|| > \epsilon\} \right] \rightarrow 0 \quad (for \ all \ \epsilon > 0).
\]
It follows by
\begin{align}
	\max_{1 \leq i \leq N} ||a_i|| \leq \sqrt{m \left(\max\left\{ \frac{\sqrt{N}}{Np}, \frac{\sqrt{N}}{N(1-p)} \right\} \right)^2}
	= \sqrt{\frac{m}{N (\min\{p, 1-p\})^2 }}
	\rightarrow 0 \quad (N \rightarrow \infty).
\end{align}
In fact, when the above condition holds, we have 
\begin{align}
	\sum_{i=1}^{N} \Exp \left[||a_i Y_{0,i}||^{2} \ind\{||a_i Y_{0,i}|| > \epsilon\} \right]
	&= \sum_{i=1}^{N} ||a_i||^2 \Exp \left[Y_{0,i}^2 \ind\{|Y_{0,i}| >\epsilon/||a_i|| \} \right] \\
	&\leq \sum_{i=1}^{N} ||a_i||^2 \max_{1 \leq i \leq N} \Exp \left[Y_{0,i}^2 \ind\{|Y_{0,i}| >\epsilon / ||a_i|| \} \right] \\
	&= \tr(AA^{\T}) \Exp \left[Y_{0,1}^2 \ind\{|Y_{0,1}| >\epsilon / \max_{1 \leq i \leq N} ||a_i|| \} \right] \\
	&= \frac{m}{p(1-p)} \Exp \left[Y_{0,1}^2 \ind\{|Y_{0,1}| >\epsilon / \max_{1 \leq i \leq N} ||a_i|| \} \right] \\
	&\rightarrow 0 \quad (N \rightarrow \infty)
\end{align}
for any $\epsilon > 0$ and Lindeberg's condition holds.
\end{proof}

\subsection{Lemmas}

\begin{lem}
\label{lem inner product} 
\begin{align}
	\wtld{z^{(k)}} \cdot \wtld{z^{(l)}} =
	\begin{dcases}
		\frac{1}{Np(1-p)} & (k=l) \\
		\frac{\rho}{Np(1-p)} & (k \neq l)
	\end{dcases}
\end{align}
holds under Assumption \ref{ass2}.
\end{lem}

\begin{proof}
Under Assumption \ref{ass2}(a), we can write that
\begin{align}
	\wtld{z_i} = 
	\begin{dcases}
		\frac{1}{Np} & (z_i = 1) \\
		-\frac{1}{N(1-p)} & (z_i = -1)
	\end{dcases}
	= \frac{(1+z_i)/2 - p}{Np(1-p)}.
\end{align}
Also we have $\sum_{i=1}^{N} \wtld{z^{(k)}_i} = 0$ and $\wtld{z^{(k)}} \cdot z^{(l)}/2 = \rho$ under Assumption \ref{ass2}(b).
Then,
\begin{align}
	\wtld{z^{(k)}} \cdot \wtld{z^{(l)}}
	&= \sum_{i=1}^{N} \wtld{z^{(k)}_i} \frac{(1+z^{(l)}_i)/2 - p}{Np(1-p)} \\
	&= \frac{1}{Np(1-p)} \sum_{i=1}^{N} \frac{\wtld{z^{(k)}_i} z^{(l)}_i}{2}
		+ \frac{1/2 - p}{Np(1-p)} \sum_{i=1}^{N} \wtld{z^{(k)}_i} \\
	&= \frac{\rho}{Np(1-p)}
\end{align}
holds.
\end{proof}

\begin{lem}
\label{g property}
Given $\Theta > 0$,
\[
	g(t; \Theta) = \frac{\phi(\Phi^{-1}(t) + \Theta)}{\phi(\Phi^{-1}(t))}
		\quad (0 < t < 1)
\]
is strictly decreasing for $t$.
\end{lem}

\begin{proof}
Putting $u = \Phi^{-1}(t)$,
\[
	g(t(u); \Theta) = \exp \left( -\frac{(u+\Theta)^2}{2} + \frac{u^2}{2} \right) 
	= \exp \left(-\Theta u - \frac{\Theta^2}{2} \right)
\]
is strictly decreasing for $u$.
Since $u$ is increasing function of $t$, $g(t; \Theta)$ is strictly decreasing for $t$.
\end{proof}

\begin{lem}[\cite{Krieger}]
\label{h property}
Suppose that $m_1, m_2$ are integers satisfying $\lfloor m_i\alpha \rfloor = m_i\alpha \ (i=1, 2)$ and $m_1 < m_2$.
Then, 
\[
	h(t) = \Fbin(\lfloor m_2\alpha \rfloor-1; \ m_2-1, \ t) - \Fbin (\lfloor m_1\alpha \rfloor-1; \ m_1-1, \ t) \quad (0 \leq t \leq 1)
\]
satisfies the following two properties:
\begin{description}
	\item[$(\mathrm{i})$] There exists $t_0 \in (0, 1)$ such that $h(t)>0 \ (t<t_0), \ h(t)<0 \ (t>t_0)$.
	\item[$(\mathrm{ii})$] $\int_{0}^{1}h(t)dt = 0$.
\end{description}
\end{lem}

\begin{proof}
For simplicity of notation, let $k_i = \lfloor m_i\alpha \rfloor - 1 = m_i\alpha -1$ and $n_i = m_i-1 \ (i = 1, 2)$.
Using the fact that the cumulative distribution function of the binomial distribution is
\begin{align}
	\Fbin(k; n, p) &= I_{1-p}(n-k, \ k+1) = \frac{B(1-p; \ n-k, \ k+1)}{B(n-k, \ k+1)} \\
	&= \int_{0}^{1-p} \frac{\theta^{n-k-1}(1-\theta)^{k}}{B(n-k, \ k+1)} d\theta,
\end{align}
where $I_x(a, b)$ is regularized beta function  (the ratio of incomplete beta function $B(x; a, b)$ to beta function $B(a, b)$), we can express that 
\begin{align}
	h(t) = \int^{1-t}_{0} \left\{ \frac{\theta^{n_2-k_2-1} (1-\theta)^{k_2}}{B(n_2-k_2, \ k_2+1)} 
		- \frac{\theta^{n_1-k_1-1} (1-\theta)^{k_1}}{B(n_1-k_1, \ k_1+1)} \right\} d\theta.
\end{align}
Differentiating by $t$ to see the behavior of $h(t)$, we have
\[
	h'(t) = - \frac{(1-t)^{n_2-k_2-1} \ t^{k_2}}{B(n_2-k_2, \ k_2+1)}
		+ \frac{(1-t)^{n_1-k_1-1} \ t^{k_1}}{B(n_1-k_1, \ k_1+1)}.
\]
A simple calculation shows that the condition for $h'(t) = 0$ is 
\[
	(1-t)^{(n_2-n_1)-(k_2-k_1)} \ t^{k_2-k_1} = \frac{B(n_2-k_2, \ k_2+1)}{B(n_1-k_1, \ k_1+1)}.
\]
Here, since both $(n_2-n_1) - (k_2-k_1)$ and $k_2 - k_1$ are positive integers, the function on the left-hand side is (a) 0 for $t=0, 1$, (b) positive for $0<t<1$, and (c) unimodal.
Thus, there are two points satisfying the above equation (if not, $h(t)$ is monotonically increasing or decreasing on $[0, 1]$, which contradicts $h(0) = h(1) = 0$).
Putting these two points as $t_1$ and $t_2 \ (0<t_1<t_2<1)$, $h(t)$ is a function that increases in $[0, t_1) \cup (t_2, 1]$ and decreases in $(t_1, t_2)$.
Therefore, together with $h(0) = h(1) =0$, there exists some $t_0 \in (t_1, t_2)$ such that $h(t) > 0 \ (0 < t < t_0), \ h(t) < 0 \ (t_0 < t < 1)$.

Considering when $\Theta=0$, Proposition \ref{prop alpha} shows that $\Pcal(0, m_1; \alpha) = \Pcal(0, m_2; \alpha) = \alpha$, and hence
\[
	0 = \Pcal(0, m_2; \alpha) - \Pcal(0, m_1; \alpha) 
	= \int_{0}^{1} h(t) g(t; 0) dt = \int_{0}^{1} h(t) dt
\]
follows.
\end{proof}

%% mybibliography
%\bibliographystyle{abbrvnat}
%%\bibliography{library,mybib,citeulike-8901247} % When using BiBTeX 
%\bibliography{../myrefs}% 文献データベースファイル
%
%\end{document}

% mybibliography
\bibliographystyle{rss}
%\bibliography{library,mybib,citeulike-8901247} % When using BiBTeX 
%\bibliography{../myrefs}% 文献データベースファイル

\end{document}